\documentclass[a4paper,12pt]{article}
\usepackage[latin1]{inputenc}
\usepackage{graphicx}
\usepackage[cmex10]{amsmath}
\usepackage[margin=2cm,nohead]{geometry}
\usepackage{cite}
\usepackage{amsfonts}
\usepackage{amssymb}
\usepackage{mathtools}
\usepackage{dsfont}
\usepackage{amsmath}
\usepackage{mathtools}
\usepackage{amsthm}
\usepackage{algorithm}
\usepackage[noend]{algpseudocode}

\usepackage{graphicx}
\usepackage{caption}
\usepackage{subcaption}

\makeatletter
\def\BState{\State\hskip-\ALG@thistlm}
\makeatother

\newcommand{\E}[1]{\mathbb{E} \left\{#1 \right\}}

\renewcommand{\P}[1]{\mathbb{P}\left\{ #1 \right\} }
\newcommand{\I}[1]{\mathds{1}\left\{ #1 \right\}}
\newcommand{\Pw}[2]{\mathbb{P}_{#1}\left\{ #2 \right\} }

\newtheorem{definition}{Definition}
\newtheorem{theorem}{Theorem}
\newtheorem{lemma}{Lemma}
\newtheorem{proposition}{Proposition}

\DeclareMathOperator*{\esssup}{ess\,sup}
\DeclareMathOperator*{\essinf}{ess\,inf}

\usepackage{pgfplots}
\pgfplotsset{compat = newest}

\usetikzlibrary{pgfplots.external}
\usepackage{filemod}

\pgfplotscreateplotcyclelist{my black white}{%
dotted, every mark/.append mark=*\\%
dotted, every mark/.append style={solid, fill=gray}, mark=square*\\%
densely dotted, every mark/.append style={solid, fill=gray}, mark=otimes*\\%
loosely dotted, every mark/.append style={solid, fill=gray}, mark=triangle*\\%
dashed, every mark/.append style={solid, fill=gray},mark=diamond*\\%
loosely dashed, every mark/.append style={solid, fill=gray},mark=*\\%
densely dashed, every mark/.append style={solid, fill=gray},mark=square*\\%
dashdotted, every mark/.append style={solid, fill=gray},mark=otimes*\\%
dasdotdotted, every mark/.append style={solid},mark=star\\%
densely dashdotted,every mark/.append style={solid, fill=gray},mark=diamond*\\%
}

\pgfplotscreateplotcyclelist{nano mark style}{%
dashed, every mark/.append style={solid, fill=gray!30}, mark=*, line width=1.1pt\\%
dashed, every mark/.append style={solid, fill=gray!30}, mark=square*, line width=1.1pt\\%
dashed, every mark/.append style={solid, fill=gray!30}, mark=triangle*, mark size=2.5pt, line width=1.1pt\\%
dashed, every mark/.append style={solid, fill=gray!30},mark=diamond*, mark size=2.5pt, line width=1.1pt\\%
dashed, every mark/.append style={solid, fill=gray!30},mark=x, mark size=3pt, line width=1.1pt\\%
}

\pgfplotscreateplotcyclelist{nano line style}{%
color = black, line width=2pt\\%
densely dashed, color = black, line width=2pt\\%
densely dotted, color = black, line width=2pt\\%
loosely dashed, color = black, line width=2pt\\%
}

\begin{document}
\sffamily
\title{Social diversity for reducing the impact\\ of information cascades on social learning}
\author{Fernando Rosas$^{1,2}$, Kwang-Cheng Chen$^{3}$ and Deniz G\"{u}nd\"{u}z$^{2}$\vspace{0.2cm}\\
$^1$ \small{Centre of Complexity Science and Department of Mathematics, Imperial College London, UK} \\
$^2$ \small{Department of Electrical and Electronic Engineering, Imperial College London, UK}\\%
$^3$ \small{Department of Electrical Engineering, University of South Florida, USA}\\
}
\date{}
\maketitle

\begin{abstract}
\noindent 
Collective behavior in online social media and networks is known to be capable of generating non-intuitive dynamics associated with crowd wisdom and herd behaviour. Even though these topics have been well-studied in social science, the explosive growth of Internet computing and e-commerce makes urgent to understand their effects within the digital society. In this work we explore how the stochasticity introduced by social diversity can help agents involved in a inference process to improve their collective performance. Our results show how social diversity can reduce the undesirable effects of information cascades, in which rational agents choose to ignore personal knowledge in order to follow a predominant social behaviour. Situations where social diversity is never desirable are also distinguished, and consequences of these findings for engineering and social scenarios are discussed.
\end{abstract}

% ---------------------------------------------------------------- %
% --- MAIN TEXT -------------------------------------------------- %
% ---------------------------------------------------------------- %

% ---------------------------------------------------------------- %
% --- Section I --- ------------------------------------------------- %
% ---------------------------------------------------------------- %
\section{Introduction}
\label{intro}

The high interconnectedness enabled by communication technologies and online media is progressively increasing the complexity of our aggregated social behaviour~\cite{bar2002complexity}. In fact, these complex dynamics were dramatically illustrated by the failure of our prediction tools in the forecast of recent political events, including the Brexit referendum and the latest US presidential election. A key open challenge is to clarify how the large amount of information that is constantly exchanged among individuals affects their decisions.

Fascinating dynamics take place when social agents engage in sequencial decision-making. For example, most people nowadays use the Internet to check other people's recommendations prior to make decisions, which enable more informed decisions thanks to the inclusion of evidence from previous experiences. Subsequent decisions are, however, heavily influenced by earlier agents, allowing misinformation or fake news to be reinforced and spread across the social network. These non-trivial \textit{social learning} dynamics are known to play a critical role in a number of key social phenomena, e.g., in the adoption or rejection of new technology, and in the formation of political opinions~\cite{easley2010networks,acemoglu2011bayesian}. Moreover, social learning also plays a key role in the context of e-commerce and digital society, e.g., in recommendation systems of online stores where users access opinions of previous customers while choosing their products~\cite{rosas2017technological,hsiao2016}. This is also the case in the emergence of viral media contents in various Internet portals, which are based on sequential actions of like or dislike.

A deep understanding of social learning dynamics is crucial for enabling robust platform design against fake news and data falsification, which is an urgent need in our modern networked society. As a matter of fact, digital misinformation was listed by the World Economic Forum (WEF) as one of the main threats to our modern society~\cite{howell2013digital}.

Social learning have been thoughtfully studied since the 90's by researchers from economics and social sciences~\cite{banerjee1992simple,bikhchandani1992theory,bikhchandani1998learning} (for modern reviews see \cite{acemoglu2011bayesian,easley2010networks}). These studies have shown that social learning is driven by two competing mechanisms. In one hand, the well-known \textit{crowd wisdom} improve the decision-making capabilities of agents within large networks, as more information becomes available to latter agents in the decision sequence. The accumulation of social experience can, on the other hand, overload agents and generate \textit{information cascades}, which pushes them to ignore their private knowledge and to adopt a predominant social behaviour. Interestingly, it has been shown that the combination of these two mechanisms can serve to provide network resilience against data falsification attacks~\cite{rosas2017social,rosas2018social}, pointing out promising possibilities for the design of resilient social learning platforms.

Motivated by the benefits that diversity can provide in biological and social systems~\cite{mathiesen2011ecosystems,santos2008social}, in this work we study how social diversity affects the learning rate in a social learning scenario. For this, we consider a network of rational agents that have diverse preferences and prior information, having some similarities to the works reported in~\cite{smith2000pathological,bala2001conformism}. Using a communication theoretical interpretation of this scenario, we show that social diversity is equivalent to additive noise in a communication channel ---which one would expect to be detrimental for the learning process. Surprisingly, our findings show that social diversity can help to avoid information cascades, introducing important improvements in the asymptotic learning performance.

The rest of this article is structured as follows. Sections~\ref{sec:2} introduces the considered social learning scenario, and develops our definition of social diversity. Section~\ref{sec:3} defines information cascades, and characterize theoretically their behaviour with respect to social diversity. Section~\ref{sec:4} presents numerical evaluations that verify the theoretical results, and finally Section~\ref{sec:5} summarizes our main conclusions.

\textbf{Notation}: uppercase letters $X$ are used to denote random variables and lowercase $x$ realizations of them, while boldface letters $\boldsymbol{X}$ and $\boldsymbol{x}$ represent vectors. Also, $\Pw{w}{X=x|Y=y} \coloneqq \P{X=x|Y=y,W=w}$ is used as a shorthand notation.

% ---------------------------------------------------------------- %
% --- Section II --------------------------------------------------- %
% ---------------------------------------------------------------- %

\section{Social learning model}
\label{sec:2}

%This section introduces the basic elements of our social learning scenario. %First, Section~\ref{sec:social_info} presents the system model, describing the sensor nodes and their interactions. Then, Section~\ref{sec:at-df} describes the behaviour of the network operator and the adversary. Finally, Section~\ref{qweqw213123} clarifies the problem that is addressed in this paper.

%This section presents the system model. For this, ... %first Section~\ref{sec:sysa} discusses the system model and basic assumptions and then Section~\ref{secA} focus in analyzing the decision rule used by agents. Finally, Section~\ref{sec:3b} develops a communication theoretic interpretation of social learning, settling the bases of the framework that is developed in the next sections.

\subsection{Preliminaries and basic assumptions}\label{sec:sysa}

Let us consider a social network composed by $N$ agents, who are engaged in a decision-making process. In this process each agent need to make a decision between two options\footnote{Although generalizations for more than two options are possible, we focus in the case of binary decisions for simplifying the presentation.}, which could correspond to a choice between two restaurants, two brands, or two political parties. It is assumed that decisions occur sequentially, and are labeled according to the order in which they take place.

The decision of the $n$-th agent, denoted as $X_n\in\{0,1\}$, is based on two sources of information (see Figure~\ref{fig:sub1}): a \textit{private signal} $S_n\in\mathcal{S}$, which is a continuous or discrete random variable that represents personal information that the $n$-th agent possesses, and \textit{social information} given by the decisions of the previous agents, denoted by $\boldsymbol{X^{n-1}} \coloneqq (X_1,\dots,X_{n-1}) \in\{0,1\}^{n-1}$.
\begin{figure}[!h]
\centering
%\begin{subfigure}{.4\textwidth}
%  \centering
  \includegraphics[width=6cm]{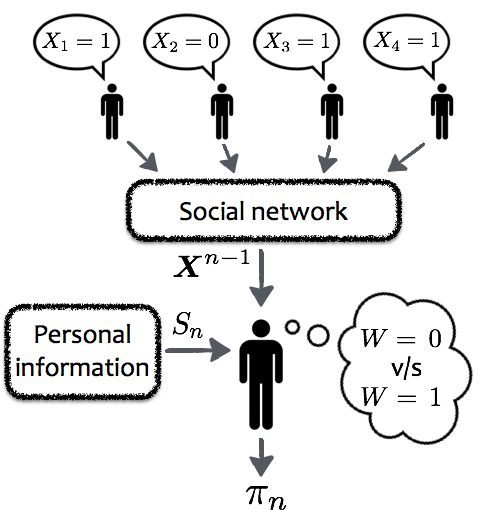}
%  \caption{Social learning scenario}
%  \label{fig:sub1}
%\end{subfigure}%
%\begin{subfigure}{.58\textwidth}
 % \centering
%  \includegraphics[width=1\linewidth]{figs/multidec3.png}
%  \caption{Communication theoretic interpretation}
%  \label{fig:sub2}
%\end{subfigure}
\caption{\small{A social learning scenario, where an agent needs to make a decision ($\pi_n$) based on personal information coming from a private signal ($S_n$) and social information ($\boldsymbol{X}^{n-1}$) coming from a social network.}}% (Right): the same scenario can be understood as a case of distributed signal processing. The depicted decoder implements the rule given by \eqref{eq:asdasd11}, combining the information provided by the private signal with the social information.}}
\label{fig:sub1}
\end{figure}

%\begin{figure}[h]
%\centering
%\begin{subfigure}{.4\textwidth}
% \centering
%  \includegraphics[width=1\linewidth]{figs/scenario.png}
%  \caption{Social learning scenario}
%  \label{fig:sub1}
%\end{subfigure}%
%\begin{subfigure}{.58\textwidth}
%  \centering
%  \includegraphics[width=1\linewidth]{figs/multidec3.png}
%  \caption{Communication theoretic interpretation}
%  \label{fig:sub2}
%\end{subfigure}
%\caption{\small{(Left): a social learning scenario, where an agent needs to make a decision ($\pi_n$) based on personal information coming from a private signal ($S_n$) and social information ($\boldsymbol{X}^{n-1}$) coming from a social network. (Right): the same scenario can be understood as a case of distributed signal processing. The depicted decoder implements the rule given by \eqref{eq:asdasd11}, combining the information provided by the private signal with the social information.}}
%\label{fig:test}
%\end{figure}

All the agents are assumed to have equivalent observation capabilities, and therefore the private signals $S_n$ are identically distributed. These signals are affected by environment conditions, which for simplicity are represented by a binary variable $W$. For the sake of tractability, we follow the existent literature in assuming that the private signals $S_n$ are conditionally independent given $W$, leaving other cases for future work. The corresponding conditional probability distributions of $S_n$ given the event $\{W=w\}$ are denoted by $\mu_w$. We further assume that no realization of $S_n$ is capable of completely determining $W$, which is equivalent to the measure theoretic notion of absolute continuity between $\mu_0$ and $\mu_1$~\cite{Loeve1978}. As a consequence of this assumption, the log-likelihood ratio of these two distributions$\mu_1$ and $\mu_0$ is well-defined and given by the logarithm of the corresponding Radon-Nikodym derivative $\Lambda_S(s) = \log \frac{d \mu_1}{d \mu_0} (s) $\footnote{When $S_n$ takes a discrete number of values then $\frac{d \mu_1}{d \mu_0} (s) = \frac{ \P{ S_n=s|W=1}}{ \P{ S_n=s|W=0}}$, while if $S_n$ is a continuous random variable with conditional p.d.f. $p(s|w)$ then $\frac{d \mu_1}{d \mu_0} (s) = \frac{ p(s|w=1) }{ p(s|w=0) }$.}. %Additionally, it is assumed that $\mu_0 \neq \mu_1$, so that $\Lambda_S(s)$ is not trivially equal to zero. %Note that, once a particular strategy is adopted, this induces a specific statistical distribution over $\boldsymbol{X}^n$.

%The log-likelyhood considered as a random variable, $\Lambda(S_n)$, plays a important role in our analysis. In fact, using the Bayes rule, the posterior probability $\P{W = 1|S_n}$ can be expressed in terms of $\Lambda(S_n)$ as
%
%\begin{equation}\label{eq:P}
%\P{W = 1|S_n} = \frac{ 1 }{ 1 + e^{ \eta - \Lambda(S_n)}} = \phi( \Lambda(S_n) - \eta)
%\end{equation}
%
%where $\eta = \log \P{W=0} - \log \P{W=1}$ is the logarithm of the ratio of the two priors and $\phi(\cdot)$ is the sigmoid function. Note that $\Lambda(S_n)\in \mathbb{R}$, and that high and low values of $\Lambda(S_n)$ correspond to values of $\P{W=1|S_n}$ close to 1 or 0, respectively.% (see Figure~\ref{fig:diagram}). 
%We denote $U_s = \sup_{s\in\mathcal{S}} \Lambda(S_n=s)$ and $L_s = \inf_{s\in\mathcal{S}} \Lambda(S_n=s)$. If any of them are unbounded, it means there exist signals that provide overwhelming evidence in favour of one of the hypothesis. If both are finite, the agent is said to have ``bounded beliefs''.

%The social information $\boldsymbol{X}^{n-1}=(X_1,\dots,X_{n-1})$ are the previous decisions that the $n$-th agent can observe in the social network. 

A \emph{strategy} is a rule for generating a decision $X_n$ based on $S_n=s$ and $\boldsymbol{X}^{n-1}$, i.e. a collection of deterministic or random functions $\pi_n:\mathcal{S}\times \{0,1\}^{n-1}\to \{0,1\}$ such that $X_n = \pi_n(S_n,\boldsymbol{X}^{n-1})$.

\subsection{Bayesian strategy, agents' preferences and prior information}
\label{secA}

Let us assume that the preferences of the $n$-th agent are encoded in an utility function $u_n(x,w)$, which determines the payoff that the agent receives when making the decision $X_n=x$ under the condition $\{W=w\}$. We consider rational agents that follow a \textit{Bayesian strategy}, which seeks to maximize their average payoff given by $\E{ u(\pi_n(S_n,\boldsymbol{X}^{n-1}),W) }$, with $\E{\cdot}$ being the expected value operator. It has been shown that the Bayesian strategy for the $n$-th agent can be expressed succinctly as \cite{rosas2017technological}
\begin{equation}\label{eq:bayes}
\frac{ \P{W = 1|S_n,\boldsymbol{X}^{n-1}} }{ \P{W = 0|S_n,\boldsymbol{X}^{n-1}} }
\mathop{\lessgtr}_{X_n=1}^{X_n=0}
e^{\nu_n}
\enspace,
\end{equation}
where $\nu_n =  \log \frac{  u_n(0,0) - u_n(0,1)  }{  u_n(1,1) - u_n(1,0)  }$ reflects the effect of the cost function. For considering agents with diverse preferences, we assume that $\nu_i$ are independent and identically distributed (i.i.d.) random variables.

Let us further consider the case where the agents have no absolute knowledge about the prior distribution of $W$. Note that because $W$ is binary, its distribution is completely determined by the value of $\P{W=1}$. Following the framework of Bayesian inference~\cite{gelman2014bayesian}, let us consider $\theta_n\in[0,1]$ to be a collection of i.i.d. random variables following a distribution $f_\theta(\theta)$ that reflects the state of knowledge of the agents about $\P{W=1}$. In particular, if the agent has complete knowledge then $f_\theta(\theta)$ is a delta centered in the true value of $\P{W=1}$ and hence $\theta_n=\P{W=1}$ for all $n$, while if agents has no information then $f_\theta(\theta)$ corresponds to an uniform distribution over $[0,1]$.

Noting that $\boldsymbol{X}^{n-1}$ depends only on $(S_1,\dots,S_{n-1})$, and therefore is conditionally independent of $S_n$, a direct application of the Bayes rule on $ \P{W = 1|S_n,\boldsymbol{X}^{n-1}} $ and $ \P{W = 0|S_n,\boldsymbol{X}^{n-1}} $ shows that \eqref{eq:bayes} can be re-written as
\begin{equation}\label{eq:asdasdA3}
\Lambda_S(S_n) + \Lambda_{\boldsymbol{X}^{n-1}}(\boldsymbol{X}^{n-1}) \mathop{\lessgtr}_{X_n=1}^{X_n=0}  \nu_n + \log \frac{\theta_n}{ 1-\theta_n}
\enspace,
\end{equation}
where $\Lambda_S(S_n)$ and $\Lambda_{\boldsymbol{X}^{n-1}}(\boldsymbol{X}^{n-1})$ are the log-likelihood ratios of $S_n$ and $\boldsymbol{X}^{n-1}$, respectively. Note that an efficient method for computing $\tau_n(\boldsymbol{X}^{n-1})$ has been reported in~\cite{rosas2017social}.

\subsection{Communication theoretic interpretation}
\label{sec:3b}

By using an adequate decision labeling, one can consider the event $\{X_n=W\}$ to be more desirable than $\{X_n\neq W\}$, or equivalently, that $u_n(1,1) \geq u_n(1,0)$ and $u_n(0,0) \geq u_n(0,1)$. The Bayesian strategy is, hence, to choose $X_n$ as similar to $W$ as possible using the information provided by $S_n$ and $\boldsymbol{X}^{n-1}$. Therefore, the decisions $\pi_n(S_n,\boldsymbol{X}^{n-1}) = X_n$ can be considered to be noisy estimations of $W$.

To further explore this perspective, let us re-formulate \eqref{eq:asdasdA3} as
\begin{equation}\label{eq:asdasd11}
\Lambda_S(S_n) + \xi_n
\mathop{\lessgtr}_{X_n=1}^{X_n=0}
\tau_n(\boldsymbol{X}^{n-1})
\enspace,
\end{equation}
where $\xi_n:=\log (1-\theta_n)/\theta_n - \nu_n$ and $\tau_n(\boldsymbol{X}^{n-1}) :=  - \Lambda_{\boldsymbol{X}^{n-1}}(\boldsymbol{X}^{n-1})$. The above can be understood as a classic signal decoder within communication theory~\cite[Section IV]{rosas2017technological}, where $\Lambda_S(S_n)$ is the decision signal and $\xi_n$ is additive noise. Moreover, $\tau_n(\boldsymbol{X}^{n-1}$ is a decision threshold that establishes the decoding rule based on a Vonoroi tessellation that divides $\mathbb{R}$ in two semi-open intervals given by $(-\infty, \tau_n(\boldsymbol{X}^{n-1}))$ and $(\tau_n(\boldsymbol{X}^{n-1}), \infty)$.

% ---------------------------------------------------------------- %
% --- Section III --------------------------------------------------- %
% ---------------------------------------------------------------- %

\section{Avoiding information cascades via noise} 
\label{sec:3}

\subsection{Local and global information cascades}

In general, the decision $\pi_n(S_n,\boldsymbol{X}^{n-1})$ is made based in complementary evidence provided by both $\boldsymbol{X}^{n-1}$ and $S_n$. The $n$-th agent is said to fall into a \emph{local information cascade} when the information conveyed by $S_n$ is not included in the decision-making process due to a dominant influence of $\boldsymbol{X}^{n-1}$. The term ``local'' is used to emphasize that this event is related to the data fusion taking place at an individual agent. The notion of local information cascade is formalized in the following definition, which is based on the notion of conditional mutual information~\cite{cover2012elements}, denoted as $I(\cdot;\cdot|\cdot)$.

\begin{definition}%\textbf{Definition (local information cascade)}: 
\label{local}
The social information $\boldsymbol{x}^{n-1}_\text{c} \in \{0,1\}^{n-1}$ generates a \emph{local information cascade} for the $n$-th agent if $I(\pi_n;S_n|\boldsymbol{X}^{n-1} = \boldsymbol{x}^{n-1}_c) = 0$.
\end{definition}

The above condition summarizes two possibilities: either $\pi_n(s,\boldsymbol{x}^{n-1}_c)$ is constant for all values of $s\in\mathcal{S}$, or there is still variability but this variability is conditionally independent of $S_n$ (e.g. in the case of stochastic strategies ---not considered in this work). In both cases, the above definition highlights the fact that the decision $\pi_n$ contains no unique information\footnote{For a rigorous definition of unique information in Markov chains c.f. \cite{rosas2016understanding}.} coming from $S_n$ when a local cascade takes place . %For the particular case of Bayesian strategies $\pi_n^b: \mathcal{S}\times \mathcal{G}_n \to \{0,1\}$ is a deterministic function, and hence the above definition states that $\boldsymbol{g}_n$ causes a local cascade if and only if $\pi(s,\boldsymbol{g}_n)$ is constant for all $s\in\mathcal{S}$.

Next we define \emph{global information cascades}, which are avalanches of local information cascades that affect all the agents after their ignition.
\begin{definition}%\textbf{Definition (global information cascade)}: 
\label{global}
The social information vector $\boldsymbol{x}^{n-1}_\text{c} \in \{0,1\}^{n-1}$ triggers a \emph{global information cascade} if $I(\pi_m:S_m|\boldsymbol{X}^{n-1} = \boldsymbol{x}^{n-1}_\text{c}) = 0$ holds for all $m\geq n$.
\end{definition}

The relationship between local and global information cascades is explored in the next section (c.f. Proposition~\ref{teo2}).

\subsection{The effect of social diversity over information cascades}

Let us first introduce $F_w(z) = \Pw{w}{\Lambda_S(S_n)+\xi_n \leq z} $ as a shorthand notation for the cumulative distribution function of $\Lambda_S(S_n)+\xi_n$ conditioned on the event $\{W=w\}$. Note that, thank to the fact that $\Lambda_S(S_1)$ and $\xi_1$ are independent random variables, one can compute $F_w(\cdot)$ as the convolution of their density functions.

\begin{lemma}\label{lemma1}
The conditional statistics of $\pi_n$ given $\boldsymbol{X}^{n-1} $ are defined by
\begin{equation}
\Pw{w}{\pi_n = 0 |\boldsymbol{X}^{n-1}=\boldsymbol{x}^{n-1}} = F_w(\tau_{n}(\boldsymbol{x}^{n-1})) .  \label{eq:tauuuw}
\end{equation}
\end{lemma}

\begin{proof}
A direct calculation shows that
\begin{equation}
\Pw{w}{\pi_1(S_1) = 0 } = \Pw{w}{ \Lambda_S(S_1) + \xi_1 < 0 } = F_w(0). \nonumber
\end{equation}
Following a similar rationale, one can find that
\begin{align}
\Pw{w}{\pi_n = 0 |\boldsymbol{X}^{n-1}=\boldsymbol{x}^{n-1}} &= \int_\mathcal{S} 
\Pw{w}{ \pi_n(s,\boldsymbol{x}^{n-1}) = 0 | S_n=s } \mu_w(s) \text{d} s  \nonumber\\
&= \int_\mathcal{S} \I{ \pi_{n} (s,\boldsymbol{x}^{n-1}) = 0 } \mu_w(s) \text{d} s \nonumber\\
&= \Pw{w}{ \Lambda_S(s) +\xi_n< \tau_{n}(\boldsymbol{x}^{n-1}) }\nonumber\\
&= F_w(\tau_{n}(\boldsymbol{x}^{n-1})) \nonumber
\enspace.
\end{align}
Above, the first equality is a consequence of the fact that $S_n$ is conditionally independent of $\boldsymbol{X}^{n-1}$ given $W=w$, while the second equality is a consequence that $\pi_n$ is a deterministic function of $\boldsymbol{X}^{n-1}$ and $S_n$, and hence becomes conditionally independent of $W$.
\end{proof}

Next, using Lemma~\ref{lemma1}, one can show that $\tau_n$ is an effective summary of the information provided by $\boldsymbol{X}^{n-1}$ that is relevant for generating the decision $\pi_n$.
\begin{lemma}\label{lemma2}
The variables $\boldsymbol{X}^{n-1} \rightarrow \tau_n \rightarrow \pi_n$ form a Markov Chain, i.e. $\tau_n$ is a sufficient statistic of $\boldsymbol{X}^{n-1}$ for predicting the decision $\pi_n$.
\end{lemma}
\begin{proof}
Using \eqref{eq:tauuuw}, one can find that
\begin{equation}
\Pw{w}{\pi_n =0| \tau_n, \boldsymbol{X}^{n-1}} =  F_w(\tau_{n}) = \Pw{w}{\pi_n=0|\tau_n} \label{eq:tauuuw0}
\enspace,
\end{equation}
and therefore the conditional independency of $\pi_n$ and $\boldsymbol{X}^{n-1}$ given $\tau_n$ is clear. 
\end{proof}

We now present a proposition that clarifies the relationship between local and global information cascades. This result extends \cite[Theorem 1]{rosas2017technological} to the current scenario.

\begin{proposition}\label{teo2}
Each local information cascade triggers a global information cascade over the social network.
\end{proposition}

\begin{proof}
Letus first note that
\begin{align}
\tau_{n+1}(\boldsymbol{X}^{n}) - \tau_n(\boldsymbol{X}^{n-1}) =& \Lambda_{\boldsymbol{X}^{n-1}}( \boldsymbol{X}^{n-1} ) - \Lambda_{\boldsymbol{X}^{n}}(\boldsymbol{X}^{n}) \nonumber\\
=&  - \Lambda_{X_n|\boldsymbol{X}^{n-1}}(X_n | \boldsymbol{X}^{n-1}) \label{eq:loc_cond}
\enspace,
\end{align}
where the conditional log-likelihood is given by
\begin{equation}
\Lambda_{X_n|\boldsymbol{X}^{n-1}}(X_n|\boldsymbol{X}^{n-1} )= \log \frac{ \Pw{1}{ X_n| \boldsymbol{X}^{n-1} } }{ \Pw{0}{X_n | \boldsymbol{X}^{n-1} }}.
\nonumber
\end{equation}

Let us consider $\boldsymbol{x}^{n-1}_\text{c}\in\{0,1\}^{n-1}$ such that it produce a local cascade in the $n$-th node. As Bayesian strategies are deterministic, local information cascades corresponds to the events where $\pi_n$ is fully determined by $\boldsymbol{X}^{n-1}$, i.e. when the probability of the event $\{\pi_n =0| \boldsymbol{X}^{n-1} = \boldsymbol{x}^{n-1}_\text{c}\} = \{X_n =0| \boldsymbol{X}^{n-1} = \boldsymbol{x}^{n-1}_\text{c}\}$ is either 0 or 1. This, in turn, implies that $\Lambda_{X_n|\boldsymbol{X}^{n-1}}(X_n|\boldsymbol{x}^{n-1}_\text{c} ) = 0$ almost surely, and therefore, conditioned on the event $\{\boldsymbol{X}^{n-1}=\boldsymbol{x}^{n-1}_\text{c}\}$ one has that
\begin{equation}
\tau_{m}(\boldsymbol{X}^{m}) = \tau_n(\boldsymbol{x}^{n-1}_\text{c}) \qquad \text{for all }m\geq n.
\end{equation}

Finally, by using \eqref{eq:tauuuw}, one can show that 
\begin{equation}
\Pw{w}{\pi_m=0|\boldsymbol{X}^m, \boldsymbol{X}^{n-1} = \boldsymbol{x}^{n-1}_\text{c}} = F_w(\tau_m(\boldsymbol{X}^m)) = F_w(\boldsymbol{x}^{n-1}_\text{c})
\end{equation}
Therefore, $\mathbb{P}_w\{\pi_m=0|\boldsymbol{X}^{n-1} = \boldsymbol{x}^{n-1}_\text{c},X_{n},\dots,X_m\} $ is also either zero or one, showing that the $m$-th agent also is affected by a local information cascade.

\end{proof}

Let us now introduce $U_s = \esssup \Lambda_S(S_n)$, $U_{\xi_n} = \esssup \xi_n$, $L_s = \essinf \Lambda_S(S_n)$ and  $L_{\xi_n} = \essinf \xi_n$ as shorthand notations for the essential supermum and infimum of $\Lambda_S(S_n)$ and $\xi_n$~\footnote{The essential supremum is the smallest upper bound of a random variable that holds almost surely, being the natural measure-theoretic extension of the notion of supremum \cite{dieudonne1976}.}. In particular, $U_s$ and $L_s$ correspond to the signals within $\mathcal{S}$ that most strongly support the hypothesis $\{W=1\}$ and $\{W=0\}$, respectively. If one of these quantities diverge, this implies that there are signals $s\in\mathcal{S}$ that provide overwhelming evidence in favour of one of the competing hypotheses. On the other hand, if $U_s$ and $L_s$ are both finite then the agents are said to have \textit{bounded beliefs}~\cite{acemoglu2011bayesian}. Similarly, when both $U_\xi$ and $L_\xi$ are finite we say agents have \textit{bounded diversity}, which implies that the diversity among priors and cost functions is not too high. Using this notions, we present the main result of this work.

\begin{theorem}\label{teo1}
Local information cascades cannot take place when agents have either unbounded beliefs or unbounded diversity.
\end{theorem}
\begin{proof}
Note first that, due to the independency between $\Lambda_S(S_n)$ and $\xi_n$, one has that
\begin{align}
U_{\text{total}} & \coloneqq \esssup \left\{ \Lambda_S(S_n) + \xi_n \right\}= U_s+U_\xi,\\
 L_{\text{total}} & \coloneqq \essinf \left\{ \Lambda_S(S_n) + \xi_n \right\}= L_s+L_\xi.
\end{align}
From this, $U_{\text{total}}$ and $L_{\text{total}}$ are unbounded if and only the agents have unbounded beliefs or unbounded diversity.

From Lemma~\ref {lemma1}, it is clear that $\pi_n$ is fully determined by $\boldsymbol{x}^{n-1} \in \{0,1\}^{n-1}$ if and only if $\tau_n(\boldsymbol{x}^{n-1})$ is such that $F_w(\tau_n(\boldsymbol{x}^{n-1}))$ is zero or 1 for $w\in\{0,1\}$. Because of the definition of $F_w$, this happens whenever $\tau_n(\boldsymbol{X}^{n-1}) \notin [L_{\text{total}},U_{\text{total}}]$, proving the Proposition.
\end{proof}

Information cascades are known to degrade the learning process, preventing the error rate of the learning process $\P{\pi_n \neq W}$ from converging to zero when the social network grows~\cite{rosas2017technological}. Therefore, Theorem~\ref{teo1} reveals a non-intuitive value of social diversity, as it can safeguard social learning from information cascades. In this way, social diversity can guarantee perfect social learning to happen asymptotically, even when agents have bounded beliefs and are hence prone to herd behaviour~\cite{rosas2017technological}. However, this benefit usually comes at the price of a slower convergence, which can be detrimental for the first agents of the decision sequence. This trade-off is explored in the next section.

% ---------------------------------------------------------------- %
% --- Section IV --------------------------------------------------- %
% ---------------------------------------------------------------- %

\section{Proof of concept}
\label{sec:4}

For illustrating the findings presented in Section~\ref{sec:3}, this section presents results of simulations of a social network following the model presented in Section~\ref{sec:2}. We considered two scenarios: one where $S_n$ are binary variables that follow a binnary symmetric channel with $\P{S_n\neq w | W=w}=1/4$, and other where $S_n$ given $\{W=w\}$ are Gaussian variables $N(\mu_w,\sigma^2)$ with $\mu_w = (-1)^{1-w}$ and $\sigma^2 = 4$. These two signal models were choosen because it is known that agent following binary signals are strongly affected by information cascades, while agents following Gaussian signals are not affected by them (for further details about these scenarios c.f. \cite[Section VI]{rosas2017technological}). For simpicity, the social diversity has been modeled considering $\xi_n$ to be i.i.d. following a Gaussian distribution $N(0,\sigma_\xi^2)$, and hence $\sigma_\xi^2$ quantifies the ``diverstiy strength'' of the social network. Each scenario was simulated $10^5$ realizations, and the statistics of the learning error rate, defined as $\P{\pi_n\neq W}$ were computed afterwards.

In agreement with Theorem~\ref{teo1}, results confirm that social learning processes can be benefited by social diversity. Figure~\ref{fig:sub21} shows how the results of a collective inference carried out by agents driven by binary private signals achieve better performance asymptotically. However, for some values of social diversity the learning rate can be rather slow, making social learning not useful for small social networks. In all the studied cases it was seen that social diversity degrades the performance of the first agents in the decision sequence; however an adequate level of diversity can introduce a fast learning rate. In contrast, as illustrated in Figure~\ref{fig:sub21} for agents following Gaussian signals, social diversity was found to be always detrimental in cases where agents have unbounded beliefs. This confirms the fact that the benefits of social diversity is to avoid information cascades, which are the main cause of poor performance of social learning in large networks~\cite{rosas2017technological}. 
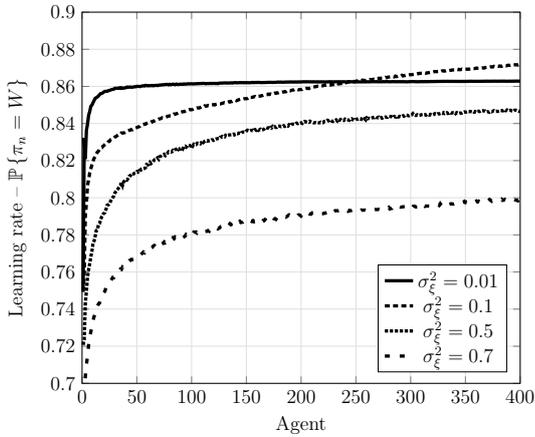
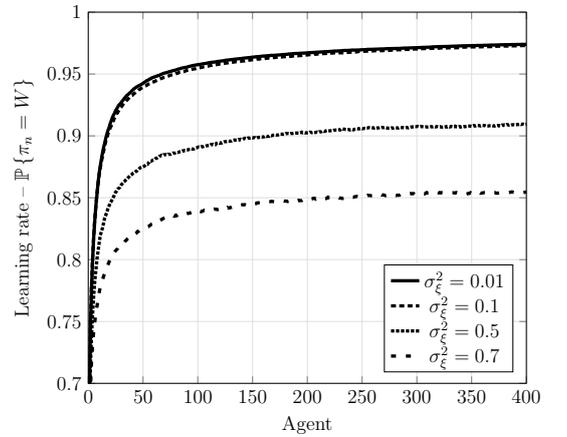
\begin{figure}[h]
\centering
\begin{subfigure}{.4\textwidth}
 \centering
\begin{tikzpicture}[scale=0.62]
  \begin{axis}
[
 %ymode = log,
 mark size=0.7pt,
 legend pos=south east,
 cycle list name = nano line style,
 xlabel={Agent},
 xmax = 400,
 xmin =0,
 ylabel={Learning rate -- $\P{\pi_n = W}$},
 ymin =0.7,
 ymax = 0.9,
 grid=both,
 minor grid style={gray!25},
 major grid style={gray!25},
 width=1.6*\linewidth,
 height=1.4*\linewidth,
no marks
] 
\addplot table[x=agent, y expr=1-\thisrow{0.01}, col sep=comma]{./bernoulli_include.csv} ; 
\addlegendentry{$\sigma_\xi^2 = 0.01$} ; 
\addplot table[x=agent, y expr=1-\thisrow{0.1}, col sep=comma]{./bernoulli_include.csv} ; 
\addlegendentry{$\sigma_\xi^2 = 0.1$} ; 
\addplot table[x=agent, y expr=1-\thisrow{0.5}, col sep=comma]{./bernoulli_include.csv} ;
\addlegendentry{$\sigma_\xi^2 = 0.5$} ;
\addplot table[x=agent, y expr=1-\thisrow{0.7}, col sep=comma]{./bernoulli_include.csv} ;
\addlegendentry{$\sigma_\xi^2 = 0.7$} ;
  \end{axis}
\end{tikzpicture}
  \caption{\small{Binary private signals}}
  \label{fig:sub21}
\end{subfigure}
\hfill
\begin{subfigure}{.4\textwidth}
%  \centering
\begin{tikzpicture}[scale=0.62]
  \begin{axis}
[
 %ymode = log,
 mark size=0.7pt,
 legend pos=south east,
 cycle list name = nano line style,
 xlabel={Agent},
 xmax = 400,
 xmin =0,
 ylabel={Learning rate -- $\P{\pi_n = W}$},
 ymin =0.7,
 ymax = 1,
 grid=both,
 minor grid style={gray!25},
 major grid style={gray!25},
 width=1.6\linewidth,
 height=1.4\linewidth,
 tick label style={/pgf/number format/fixed},
 no marks
] 
\addplot table[x=agent, y expr=1-\thisrow{0.01}, col sep=comma]{./gaussians_include.csv} ; 
\addlegendentry{$\sigma_\xi^2 = 0.01$} ; 
\addplot table[x=agent, y expr=1-\thisrow{0.1}, col sep=comma]{./gaussians_include.csv} ; 
\addlegendentry{$\sigma_\xi^2 = 0.1$} ; 
\addplot table[x=agent, y expr=1-\thisrow{0.5}, col sep=comma]{./gaussians_include.csv} ;
\addlegendentry{$\sigma_\xi^2 = 0.5$} ;
\addplot table[x=agent, y expr=1-\thisrow{0.7}, col sep=comma]{./gaussians_include.csv} ;
\addlegendentry{$\sigma_\xi^2 = 0.7$} ;
  \end{axis}
\end{tikzpicture}
  \caption{\small{Gaussian private signals}}
  \label{fig:sub22}
\end{subfigure}
\caption{\small{Social learning rate for agents following binary or Gaussian private signals, under various levels of social diversity ($\sigma_\xi^2$). Social networks that follow binary signals are vulnerable to information cascades, and hence a non-zero social diversity improve their asymptotic learning rate. In contrast, social networks that follow Gaussian signals are inmune to information cascades, and hence social diversity have a purely detrimental effect.}}
\label{fig:test}
\end{figure}

The different effect that social diversity has over agents located at different positions in the inference process is further illustrated by Figure~\ref{fig:3}. We found that, for each agent, there exists an optimal level of social diversity that reduces the effect of information cascades without introducing too much noise. Agents located in the first places of the decision sequence are always affected negatively by social diversity, and hence for them is optimal to have $\sigma_\xi^2=0$.
\begin{figure}[h]
  \centering
\begin{tikzpicture}[scale=0.62]
  \begin{axis}
[
 %ymode = log,
 mark size=0.7pt,
 legend pos=south west,
 cycle list name = nano line style,
 xlabel={Social diversity $\sigma_\xi^2$},
 xmax = 0.45,
 xmin =0.05,
 ylabel={Learning rate -- $\P{\pi_n = W}$},
 ymin =0.75,
 ymax = 0.91,
 grid=both,
 minor grid style={gray!25},
 major grid style={gray!25},
 width=0.7\linewidth,
 height=0.55\linewidth,
 tick label style={/pgf/number format/fixed},
 no marks
] 
\addplot table[x=agent, y expr=1-\thisrow{10}, col sep=comma]{./bernoulli_fixed_selected.csv} ; 
\addlegendentry{10-th agent} ; 
\addplot table[x=agent, y expr=1-\thisrow{100}, col sep=comma]{./bernoulli_fixed_selected.csv} ; 
\addlegendentry{100-th agent} ; 
\addplot table[x=agent, y expr=1-\thisrow{400}, col sep=comma]{./bernoulli_fixed_selected.csv} ;
\addlegendentry{400-th agent} ;
\addplot table[x=agent, y expr=1-\thisrow{800}, col sep=comma]{./bernoulli_fixed_selected.csv} ;
\addlegendentry{800-th agent} ;
  \end{axis}
\end{tikzpicture}
\caption{An optimal level of social diversity exist that can improve the social learning performance of agents located late in the inference process. However desirable, this better performance of late agents comes at the expense of a detrimental effect to the first agents.}
\label{fig:3}
\end{figure}
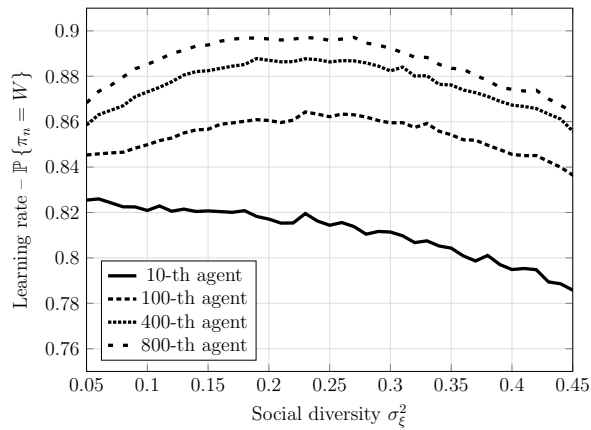

% ---------------------------------------------------------------- %
% --- Section V --------------------------------------------------- %
% ---------------------------------------------------------------- %

\section{Conclusion}
\label{sec:5}

This paper aims to undestand how social learning is affected when it is pursued by a diverse population. Our scenario considered rational agents with heterogeneous preferences, as encoded by their utility functions, and diverse prior information about the target variable. A communication theoretic analysis showed that this kind of social diversity is equivalent to additive noise in a communication channel. However, it was found that an unbounded social diversity prevent information cascades and, hence, introduces important improvements into the asymptotic social learning rate that can be achieved by a population. Social learning is, therefore, one of those rare cases where noise can improve the overall performance.

To understand how can noise be beneficial, let us point out that rational social agents maximize their individual performance while ignoring the consequences of their actions on the aggregated behaviour. This selfish quality of the agent's behavior makes their actions locally optimal while being globaly suboptimal. In this context, the heterogeneity introduced by social diversity makes the decisions of each agent less informative to others. This generates a reduced social pressure that, in turn, prevents information cascades and herd behaviour, introducing great improvements in the asymptotic social learning performance.

It is to be noted that the benefits of social diversity are only experienced by agents that are prone to information cascades. Therefore, social diversity is not beneficial, e.g., for agents with unbounded beliefs. However, in most applications agent's beliefs are bounded, either because their signals information content is limited or because the signals themselves are bounded. The latter is the case in most engineering applications, e.g. in the scenario studied in \cite{rosas2017social}.

Finally, it is important to remark that social diversity provides benefits to the latter agents in the decision sequence, while degrading the performance of the first agents. Therefore, social diversity might in general be detrimental for the performance of social learning in small networks.

%As a final note, As sensory signals of electronic devices are ultimately processed digitally, the number of different signals that an agent can obtain are finite and hence their supremum is always finite. Therefore, in the sequel we asume that both $L_s$ and $U_s$ are finite. Using these notions, the following proposition provides a characterization for local information cascades.

\section*{Acknowledgements}

Fernando Rosas is supported by the European Union's H2020 research and innovation programme, under the Marie Sk\l{}odowska-Curie grant agreement No. 702981.

% ---------------------------------------------------------------- %
% --- REFERENCES ------------------------------------------------- %
% ---------------------------------------------------------------- %


\end{document}